%% file: NLDarXiv.tex
\documentclass[letterpaper,10pt,conference]{IEEEtran}
\usepackage{graphicx,ifthen}
\usepackage{psfrag,amssymb,amsthm}
\usepackage[cmex10]{amsmath}
\usepackage{cite,flushend,color,pst-plot,stfloats,pst-sigsys}
\usepackage{pstricks-add}

\newtheorem{thm}{Theorem}
\newtheorem{lem}{Lemma}
\newtheorem{cor}{Corollary}
\theoremstyle{definition}
\newtheorem{definition}{Definition}

\title{Some Results on the Information Loss in Dynamical Systems}

\author{\IEEEauthorblockN{Bernhard C. Geiger\IEEEauthorrefmark{1}, Gernot Kubin\IEEEauthorrefmark{1}
\IEEEauthorblockA{\IEEEauthorrefmark{1}Signal Processing and Speech Communication Laboratory, Graz University of Technology, Austria}
$\{$geiger,gernot.kubin$\}$@tugraz.at}}

\begin{document}
\newcounter{myTempCnt}
\input{abbrevations_processing.tex}

\maketitle

\begin{abstract}
In this work we investigate the information loss in (nonlinear) dynamical input-output systems and provide some general results. In particular, we present an upper bound on the information loss rate, defined as the (non-negative) difference between the entropy rates of the jointly stationary stochastic processes at the input and output of the system.

We further introduce a family of systems with vanishing information loss rate. It is shown that not only linear filters belong to that family, but -- under certain circumstances -- also finite-precision implementations of the latter, which typically consist of nonlinear elements.
\end{abstract}

\section{Introduction}\label{sec:intro}
Transmission and processing of information is the primary concern in many fields of communications, signal processing, and machine learning. The typical impairments considered in these contexts are noise and interference, incomplete data sets, and coarse observations, eliciting both information-theoretic and energy-centered analyses. In contrary, the effect of deterministic input-output systems on the information content, i.e., the entropy rate, of a signal has not yet been thoroughly analyzed. Still, nonlinear dynamical systems -- capable of changing information content -- are omnipresent in communication systems in the roles of high-power amplifiers or frequency mixers. Another example is the energy detector, a low-complexity receiver architecture for wireless communications. To obtain a better understanding of the effects of these system components, an information-theoretic treatment is essential.

In this paper, we establish a framework for analyzing the effects of discrete-time dynamical systems with a finite-dimensional state vector on the entropy rate of a signal. While the analysis of continuous-valued stochastic processes will be left for future work, here we focus on (jointly) stationary input and output processes taking values from countable alphabets. 

The data processing inequality (DPI,~\cite[pp.~35]{Cover_Information2}) states that the entropy of a discrete random variable (RV) cannot increase by passing the RV through a static nonlinearity. It was shown that the same result holds for entropy rates of jointly stationary stochastic processes on finite alphabets, both for static nonlinearities~\cite{Watanabe_InfoLoss} and general dynamical systems~\cite{Pinsker_InfoStab}.
Continuous-valued processes passing through linear filters were already analyzed by Shannon in terms of differential entropy rates~\cite{Shannon_TheoryOfComm,Papoulis_Probability}, which in our opinion are not adequate measures of information loss, cf. Section~\ref{sec:linfilt}. The conditional entropy, used to characterize the information lost by passing a continuous RV through a static nonlinearity~\cite{Geiger_ISIT2011arXiv} or by multiplying two integers~\cite{Pippenger_MultLoss}, appears to be more appropriate.

We start by defining the information loss rate in Section~\ref{sec:problem} and show that this quantity is equal to the difference between the entropy rates of the input and output processes. This choice establishes the DPI for dynamical systems in Section~\ref{sec:infoLossRate}, stating that the information loss rate is non-negative. This result is then complemented by an upper bound that can be evaluated easily. In Section~\ref{sec:partInv} we introduce a family of dynamical systems for which we show that the information loss rate vanishes. This family not only comprises a large class of stable linear filters (see Section~\ref{sec:linfilt}), but also their finite-precision counterparts, commonly used in digital signal processing. Aside from the latter, Section~\ref{sec:examples} discusses some other examples illustrating our theoretical results. 

This document is an extended version of a paper submitted to an IEEE conference.

\section{Problem Statement \& Preliminaries}
\label{sec:problem}
We consider a discrete-time regular two-sided stationary stochastic process $\Xvec$ taking values from a countable set $\dom{X}$. Let $X_n$ denote the RV of the $n$-th sample and let $X_k^n=(X_k,X_{k+1},\dots,X_n)$, thus $\Xvec=X_{-\infty}^\infty$. For the actual value of $X_n$ we write $x_n$. 
We further consider another countable set $\dom{Y}$ which needs not be identical to $\dom{X}$. Let $\ent{X_n}$ denote the zeroth-order entropy of $X_n$ and let $\entrate{X}=\limn\frac{1}{n}\ent{X_1^n}$ denote the entropy rate of $\Xvec$. The restriction to countable sets ensures that entropies and entropy rates are well-defined. 

The following class of dynamical systems is treated in this work:
\begin{definition}[Finite-Dimensional Dynamical System]\label{def:function}
 Let $Y_n=f(X_{n-N}^n,Y_{n-M}^{n-1})$, $0\leq M,N<\infty$, be the RV of the $n$-th output sample of a dynamical system with a finite-dimensional state vector subject to the input process $\Xvec$. Here, $f{:}\; \dom{X}^{N+1}\times\dom{Y}^M\to\dom{Y}$ is a function such that the sequence of output samples, $Y_n$, constitutes a two-sided stochastic process $\Yvec$ jointly stationary with $\Xvec$.
\end{definition}

\begin{definition}[Information Loss Rate]\label{def:infoloss}
 Let $\Xvec$ and $\Yvec$ be jointly stationary processes on countable sets related as in Definition~\ref{def:function}. The average information lost per sample is given by the conditional entropy rate
\begin{equation}
 \entrate{X|Y} = \limn \frac{1}{n} \ent{X_1^n|Y_1^n}.
\end{equation}
\end{definition}

Characterizing the information loss as a conditional entropy rate is quite intuitive: The conditional entropy rate denotes the average number of bits per sample unknown about the input sequence after observing the output sequence; i.e., the average information lost per sample by passing the sequence through the system in question.

Before proceeding with the analysis, we will introduce two Lemmas:

\begin{lem}\label{lem:addingDepRVs} 
 For any set of discrete RVs $Z_1^n$ and any function $f(Z_k,Z_l,\dots)$, $1\leq k,l, \dots \leq n$, the following holds:
\begin{equation}
 \ent{Z_1^n,f(Z_k,Z_l,\dots)} = \ent{Z_1^n}
\end{equation}
\end{lem}

\begin{proof}
See~\cite[Prob.~2.4]{Cover_Information2}.
\end{proof}

\begin{lem}\label{lem:finvsinf}
 Let $\Xvec$ and $\Yvec$ be jointly stationary stochastic processes on countable sets. Then, for $M<\infty$,
\begin{IEEEeqnarray}{RCL}
 \entrate{X} 
&=&\limn \frac{1}{n} \ent{X_1^n|Y_1^M}=\limn \frac{1}{n} \ent{X_1^n,Y_1^M}.\notag
\end{IEEEeqnarray}
\end{lem}

\begin{proof}
 Clearly,
\begin{equation}
 \ent{X_1^n|Y_1^M} \leq \ent{X_1^n} \leq \ent{X_1^n,Y_1^M}
\end{equation}
for all $n$, thus also in the limit. Now, since $\ent{X_1^n,Y_1^M}=\ent{X_1^n|Y_1^M}+\ent{Y_1^M}$ and since all involved entities are non-negative,
\begin{IEEEeqnarray}{RCL}
\entrate{X} & \leq& 
\limn \frac{1}{n} \ent{X_1^n|Y_1^M}+ \underbrace{\limn \frac{1}{n}\ent{Y_1^M}}_{\to 0}.\label{eq:lemma}
\end{IEEEeqnarray}
Thus in the limit the upper and lower bound are equal and the proof is completed.
\end{proof}

Since the input and output alphabets of the dynamical systems can be countable, it may occur that the entropy of a single sample becomes infinite. Yet, by the maximum entropy property of the uniform distribution,
\begin{equation}
 \ent{Y_1^M}\leq M\ent{Y}\leq \lim_{|\dom{Y}|\to\infty} M\log |\dom{Y}|
\end{equation}
which approaches infinity at a slower rate than $\limn n$. Thus the term on the right in~\eqref{eq:lemma} approaches zero even for processes $\Yvec$ with infinite zeroth-order entropy or infinite entropy rate.

\section{Information Loss Rate in Dynamical Systems}\label{sec:infoLossRate}
In this Section, which comprises the main contribution of this work, we present some general results on the information loss rate induced by a system satisfying Definition~\ref{def:function}. We will start by proving a Theorem which essentially states that the information loss rate is identical to the difference of entropy rates:

\begin{thm}\label{thm:ILisED}
 Let $\Xvec$ and $\Yvec$ be jointly stationary processes on countable sets related as in Definition~\ref{def:function}. Then, the information loss rate is given by the difference of entropy rates:
\begin{equation}
 \entrate{X|Y} = \entrate{X}-\entrate{Y}
\end{equation}
\end{thm}

\begin{proof}
While the proof for static functions (i.e., $M=N=0$) is relatively simple~\cite{Watanabe_InfoLoss}, for dynamical systems we have to show that
\begin{IEEEeqnarray}{RCL}
 \entrate{X|Y} &=& \limn \frac{1}{n} \left( \ent{X_1^n,Y_1^n}-\ent{Y_1^n}\right)\\
&=& \limn \frac{1}{n} \ent{X_1^n} - \limn \frac{1}{n}\ent{Y_1^n}
\end{IEEEeqnarray}
i.e., that
\begin{equation}
 \limn \frac{1}{n} \ent{X_1^n,Y_1^n} = \limn \frac{1}{n} \ent{X_1^n}.
\end{equation}
Consider that, for $n>\max\{M,N\}$
\begin{IEEEeqnarray}{RCL}
 \ent{X_1^n,Y_1^n} &=& \ent{Y_n,X_1^n,Y_1^{n-1}}\\
 &=& \ent{f(X_{n-N}^n,Y_{n-M}^{n-1}),X_1^n,Y_1^{n-1}}\\
 &\stackrel{(a)}{=}& \ent{X_1^n,Y_1^{n-1}}
\end{IEEEeqnarray}
where $(a)$ is due to Lemma~\ref{lem:addingDepRVs}. By repeated application,
\begin{equation}
 \ent{X_1^n,Y_1^n} = \ent{X_1^n,Y_1^{\max\{M,N\}}}.
\end{equation}
Since this holds for all $n>\max\{M,N\}$, it also holds in the limit and with Lemma~\ref{lem:finvsinf} we obtain
\begin{equation}
\limn \frac{1}{n}\ent{X_1^n,Y_1^{\max\{M,N\}}} = \limn \frac{1}{n} \ent{X_1^n}
\end{equation}
and thus
\begin{equation}
 \entrate{X|Y} = \entrate{X}-\entrate{Y}.
\end{equation}
This completes the proof.
\end{proof}

The significance of this Theorem lies in the fact that the information loss can be inferred by comparing the entropy rates of the input and output processes. Note that the same does not hold for differential entropy rates, as we will argue in Section~\ref{sec:linfilt}.

By the non-negativity of the conditional entropy rate the following Corollary to Theorem~\ref{thm:ILisED} shows that the entropy rate of the system output cannot be larger than the entropy rate of the system input. This result, originally stated in~\cite{Pinsker_InfoStab} for finite alphabets, further justifies our intuitive definition of information loss:

\begin{cor}[DPI for Dynamical Systems]\label{cor:DPIFinMem}
Let $\Xvec$ and $\Yvec$ be jointly stationary processes on countable sets related as in Definition~\ref{def:function}. Then, the entropy rate of the output process $\Yvec$ cannot be larger than the entropy rate of the input process $\Xvec$, i.e.,
\begin{equation}
 \entrate{Y}\leq\entrate{X}.
\end{equation}
\end{cor}

Generally, the computation of entropy rates is a non-trivial problem, where closed-form solutions exist only for simple processes (e.g., Markov chains). Since functions of stochastic processes rarely allow such a simplified treatment, the availability of bounds is of vital importance. We will thus present an upper bound on the information loss rate, which is simple to evaluate:
\begin{thm}[Upper Bound]\label{thm:UBInfoLossRate}
 Let $\Xvec$ and $\Yvec$ be jointly stationary processes on countable sets related as in Definition~\ref{def:function}. Then, the information loss rate is bounded by
\begin{equation}
 \entrate{X|Y} \leq \max_{(x,\theta)\in\dom{X}\times\dom{T}}\log|f^{-1}_\theta[f_\theta(x)]|
\end{equation}
where $\dom{T}=\dom{X}^N\times\dom{Y}^M$, $\theta\in\dom{T}$ are the possible values of the RV $\Theta_n=\{X_{n-N}^{n-1},Y_{n-M}^{n-1}\}$, and $f^{-1}_\theta[\cdot]$ denotes the preimage under $f_\theta$, an instantiation of the function $f_{\Theta_n}(\cdot)=f(\cdot,\Theta_n)$.
\end{thm}

\begin{IEEEproof}
 \begin{IEEEeqnarray}{RCL}
  \entrate{X|Y} &=& \limn \frac{1}{n} \left(\ent{X_1^n,Y_1^n}-\ent{Y_1^n}\right)\\
&\stackrel{(a)}{=}& \limn \frac{1}{n} \left( \sumin \ent{X_i,Y_i|X_1^{i-1},Y_1^{i-1}} \right.\notag\\
&&-\left.\sumin \ent{Y_i|Y_1^{i-1}} \right)\\
&\stackrel{(b)}{\leq}& \limn \frac{1}{n} \left( \sumin \ent{X_i,Y_i|X_1^{i-1},Y_1^{i-1}} \right.\notag\\
&&-\left.\sumin \ent{Y_i|X_1^{i-1},Y_1^{i-1}} \right)\\
&=& \limn \frac{1}{n} \sumin \ent{X_i|X_1^{i-1},Y_1^{i}}\label{eq:suminproof}
 \end{IEEEeqnarray}
where $(a)$ is due to the chain rule of entropy and $(b)$ is due to the fact that conditioning reduces entropy. The expression under the sum in~\eqref{eq:suminproof} is a non-negative decreasing sequence in $i$ and thus has a limit. We use the Ces\'aro mean~\cite[Thm.~4.2.3]{Cover_Information2} and obtain
\begin{IEEEeqnarray}{RCL}
 \entrate{X|Y} &\leq& \limn \ent{X_n|X_1^{n-1},Y_1^{n}}\\
 &\leq& \ent{X_n|X_{n-N}^{n-1},Y_{n-M}^{n}}\\
 &=& \ent{X_n|Y_n,\Theta_n}.
\end{IEEEeqnarray}
We now replace $Y_n=f(X_n,\Theta_n) = f_{\Theta_n}(X_n)$, where we treat the collection of all previous RVs influencing $Y_n$ as a (random) parameter $\Theta_n$ of the function. This approach lets us interpret the dynamical system as a parameterized static system $f_{\Theta_n}{:}\;\dom{X}\to\dom{Y}$, where we let $\Theta_n$ take values $\theta$ from $\dom{T}=\dom{X}^N\times\dom{Y}^M$. We thus continue
\begin{IEEEeqnarray}{RCL}
 \entrate{X|Y} &\leq& \ent{X_n|f_{\Theta_n}(X_n),\Theta_n}\\
&=&  \sum_{(x,\theta)\in\dom{X}\times\dom{T}} \ent{X_n|f_\theta(x),\theta} \Prob{X_n=x,\Theta_n=\theta}\notag\\
&\stackrel{(c)}{\leq}& \sum_{(x,\theta) \in\dom{X}\times\dom{T}} \log |f^{-1}_\theta[f_\theta(x)]|  \Prob{X_n=x,\Theta_n=\theta}\notag\\
&\leq& \max_{(x,\theta)\in\dom{X}\times\dom{T}} \log |f^{-1}_\theta[f_\theta(x)]|
\end{IEEEeqnarray}
where $(c)$ is due to conditioning and the maximum entropy property of the uniform distribution over an alphabet size equal to the cardinality of the preimage under $f_\theta$. Maximizing over all possible $x$ and parameter values $\theta$ completes the proof.
\end{IEEEproof}
This result can be interpreted as relating the information loss rate of a dynamical system to the information loss rate induced by a static function. In particular, we let the static function be parameterized by previous input and output values taking effect on $Y_n$ and upper bound the information loss rate by the maximum cardinality of the preimage under $f_\theta$. While this upper bound may be rather conservative, it is particularly simple to evaluate if the system function from Definition~\ref{def:function} is available. We will illustrate the use of this result in Section~\ref{ssec:Hammerstein}.

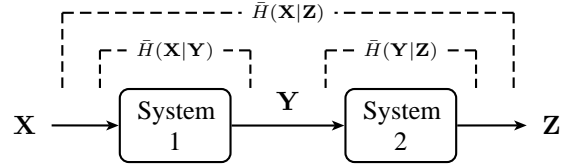
\begin{figure}[t]
 \centering
\begin{pspicture}[showgrid=false](1,1.5)(8,4)
	\psset{style=RoundCorners}
 	\pssignal(1,2){x}{$\Xvec$}
	\psfblock[framesize=1.5 1](3,2){d}{\parbox{\psunit}{\centering System 1}}
	\psfblock[framesize=1.5 1](6,2){c}{\parbox{\psunit}{\centering System 2}}
	\pssignal(8,2){y}{$\Zvec$}
  \nclist[style=Arrow]{ncline}[naput]{x,d,c $\Yvec$,y}
	\psline[style=Dash](2,3)(4,3)
	\psline[style=Dash](2,2.5)(2,3)
	\psline[style=Dash](4,2.5)(4,3)
	\psline[style=Dash](5,3)(7,3)
	\psline[style=Dash](5,2.5)(5,3)
	\psline[style=Dash](7,2.5)(7,3)
	\psline[style=Dash](1.5,3.5)(7.5,3.5)
	\psline[style=Dash](1.5,2.5)(1.5,3.5)
	\psline[style=Dash](7.5,2.5)(7.5,3.5)
 	\rput*(3,3){\scriptsize\textcolor{black}{$\entrate{X|Y}$}}
	\rput*(6,3){\scriptsize\textcolor{black}{$\entrate{Y|Z}$}}
	\rput*(4.5,3.5){\scriptsize\textcolor{black}{$\entrate{X|Z}$}}
\end{pspicture}
 \caption{Cascade of systems}
 \label{fig:cascade}
\end{figure}

Finally, we present a result about the cascade of systems (see Fig.~\ref{fig:cascade}):
\begin{thm}[Cascading Systems]\label{thm:cascade}
 Let $\Xvec$, $\Yvec$, and $\Zvec$ be jointly stationary stochastic processes on countable sets $\dom{X}$, $\dom{Y}$, and $\dom{Z}$, respectively, where $\Yvec$ is generated by passing $\Xvec$ through a system satisfying Definition~\ref{def:function}, and $\Zvec$ is generated by passing $\Yvec$ through another such system. Then, the information loss rate induced by the cascade is
\begin{equation}
 \entrate{X|Z} = \entrate{X|Y} + \entrate{Y|Z}.
\end{equation}
\end{thm}

\begin{proof}
By using Theorem~\ref{thm:ILisED} $\entrate{X|Z}$ can be written as
\begin{eqnarray}
 \entrate{X|Z} &=& \entrate{X} - \entrate{Z}\\
&=& \entrate{X} - \entrate{Y} + \entrate{Y} - \entrate{Z}\\
&=& \entrate{X|Y} + \entrate{Y|Z}.
\end{eqnarray}
\end{proof}

\section{Partially Invertible Systems}
\label{sec:partInv}
We now impose an additional restriction on the system function in Definition~\ref{def:function}. This additional restriction defines a family of systems for which the information loss rate can be shown to vanish.

\begin{definition}[Partially invertible system]\label{def:partInv}
A system satisfying Definition~\ref{def:function} is \emph{partially invertible} if there exists a function $\finv{:}\; \dom{X}^N\times\dom{Y}^{M+1}\to\dom{X}$ such that
\begin{equation}
 X_n = \finv(X_{n-N}^{n-1},Y_{n-M}^n)=\finv(Y_n,\Theta_n)=f_{\Theta_n}^{-1}(Y_n).
\end{equation}
In other words, a system is partially invertible if its parameterized static function $f_{\Theta_n}$ is invertible for all possible parameter values $\theta\in\dom{T}$.
\end{definition}

We will now argue that for this class of systems the information loss rate vanishes. We start by showing that the total information loss for a finite-length input sequence $X_1^K$ after observing an output sequence $Y_1^K$ of the same length remains bounded independently of the sequence length:

\begin{thm}\label{thm:InfoLossPI}
 Let $X_1^K$ and $Y_1^K$, $K>\max\{M,N\}$, be two finite-length sequences of jointly stationary processes $\Xvec$ and $\Yvec$ on countable sets $\dom{X}$ and $\dom{Y}$, respectively, where $\Yvec$ is generated by passing $\Xvec$ through a partially invertible system. Then, the information loss becomes
\begin{equation}
 \ent{X_1^K|Y_1^K} = \ent{X_1^{\max\{M,N\}}|Y_1^K}.
\end{equation}
\end{thm}

\begin{proof}
 We start by noticing that $\ent{X_1^K|Y_1^K}=\ent{X_1^K,Y_1^K}-\ent{Y_1^K}$ and
\begin{IEEEeqnarray}{RCL}
  \ent{X_1^K,Y_1^K}&=&\ent{X_K,X_1^{K-1},Y_1^K}\\
&=&\ent{\finv(X_{K-N}^{K-1},Y_{K-M}^{K}),X_1^{K-1},Y_1^K}\notag\\
&\stackrel{(a)}{=}& \ent{X_1^{K-1},Y_1^K}
\end{IEEEeqnarray}
where $(a)$ is due to Lemma~\ref{lem:addingDepRVs}. Repeating this step a number of times yields
\begin{equation}
  \ent{X_1^K,Y_1^K} = \ent{X_1^{\max\{M,N\}},Y_1^K}.
\end{equation}
Subtracting $\ent{Y_1^K}$ completes the proof.
\end{proof}
Note that even though $f_{\Theta_n}$ is invertible for all parameter values $\theta$, this does only mean that $\ent{X_n|f_{\Theta_n}(X_n),{\Theta_n}}=0$, while $\ent{X_n|f_{\Theta_n}(X_n)}\geq 0$. This corresponds to the statement of Theorem~\ref{thm:InfoLossPI}, where for $n<\max\{M,N\}$ ${\Theta_n}$ has to be considered unknown. It is also important to note that $\ent{X_1^K|Y_1^K}\neq\ent{X_1^K}-\ent{Y_1^K}$. While the information loss rate is equal to the difference of entropy rates (cf. Theorem~\ref{thm:ILisED}), it does not hold generally that the difference of joint entropies is equal to the joint conditional entropy. 

We will now make use of this result in proving that partially invertible systems have a vanishing information loss rate:

\begin{cor}\label{cor:parInfNoLoss}
 Let $\Xvec$ and $\Yvec$ be jointly stationary processes on countable sets related as in Theorem~\ref{thm:InfoLossPI}. Then, the information loss rate induced by passing the process $\Xvec$ through the system vanishes, i.e.,
\begin{equation}
 \entrate{X|Y}=0.
\end{equation}
\end{cor}

\begin{proof}
We provide two proofs for this Corollary. For the first, note that irrespective of $\theta$ the inverse function $f_\theta^{-1}$ always exists by Definition~\ref{def:partInv}. With Theorem~\ref{thm:UBInfoLossRate} this immediately leads to $\entrate{X|Y}=0$.

For the second proof we note that Theorem~\ref{thm:InfoLossPI} holds for all $K$, thus also in the limit. With Definition~\ref{def:infoloss} we can therefore write the information loss rate as:
\begin{IEEEeqnarray}{RCL}
 \entrate{X|Y} &=& \limn \frac{1}{n}\ent{X_1^n|Y_1^n}\\
&=& \limn \frac{1}{n}\ent{X_1^{\max\{M,N\}}|Y_1^n}\\
&\leq& \limn \frac{1}{n}\ent{X_1^{\max\{M,N\}}}=0
\end{IEEEeqnarray}
by similar arguments as in the proof of Lemma~\ref{lem:finvsinf}.
\end{proof}

An immediate consequence of this important Corollary is that, except for the initial samples $X_1^{\max\{M,N\}}$ after starting the observation of $\Yvec$ (cf. Theorem~\ref{thm:InfoLossPI}), the remaining information of the input process can be recovered by observing the output process. Note that this not necessarily means that the input process can be reconstructed perfectly, even if reconstruction errors are allowed in the first $\max\{M,N\}$ samples. An illustrative example for this fact will be given in Section~\ref{ssec:multipliers}.

\section{The Case of Linear Filters}
\label{sec:linfilt}
It is interesting to note that an important subclass of discrete-time stable causal linear filters falls in the category of partially invertible systems, as long as the input and output alphabets are countable. An example where the latter condition is satisfied is given if the input process and the coefficients take values from the field of rational numbers. This subclass, powerful enough to cover most applications~\cite{Diniz_DSP}, comprises filters with a finite-dimensional state vector described by constant-coefficient difference equations:
\begin{equation}
 Y_n=\sum_{k=0}^N b_kX_{n-k} + \sum_{l=1}^M a_lY_{n-l} \label{eq:linfilt}
\end{equation}
As noted in~\cite{Papoulis_SpectralEstimation}, stability of the filter guarantees that for a stationary input process the output process is stationary and that Definition~\ref{def:function} applies. By rearranging the terms in~\eqref{eq:linfilt} it can be verified that this subclass of linear systems satisfies the definition of partially invertible systems and, thus, has a vanishing information loss rate.

It is noteworthy that this property is independent of the minimum-phase property (cf.~\cite[pp.~280]{Oppenheim_DiscreteSigProc}) of linear filters, which ensures that the filter has a stable and causal inverse. Indeed, for filters which are not minimum-phase, the partial inverse function $\finv$ used in Definition~\ref{def:partInv} describes a causal, but unstable linear filter. As a consequence, to an arbitrary stationary stochastic input process, the inverse filter described by $\finv$ may respond with a non-stationary output process; however, the response to $\Yvec$ will be $\Xvec$.

A signal space model may effectively illustrate these considerations: Let $\dom{X}^\infty$ and $\dom{Y}^\infty$ be the spaces of stationary input and output processes $\Xvec$ and $\Yvec$, respectively, and let $F\{\cdot\}$ be the (linear) operator mapping each element of $\dom{X}^\infty$ to $\dom{Y}^\infty$. By restricting our attention to regular stochastic processes, i.e., processes which cannot have periodic components, the operator $F\{\cdot\}$ is injective. As a consequence, for each element of $\dom{Y}^\infty$ there exists at most one element in $\dom{X}^\infty$ such that $\Yvec=F\{\Xvec\}$. Note, however, that there are stationary stochastic processes in $\dom{Y}^\infty$ which are not images of elements in $\dom{X}^\infty$. Only if $F\{\cdot\}$ is such that it describes a stable, causal minimum-phase system, i.e., has a \emph{stable} and causal inverse, $\dom{Y}^\infty$ contains only images of elements from $\dom{X}^\infty$.

This complements a result already introduced by Shannon~\cite{Shannon_TheoryOfComm}, which states that the change in \emph{differential} entropy rate caused by stable, causal linear filtering of continuous-valued stationary processes is independent of the process statistics. In particular, for a linear filter with frequency response $G(\e{\jmath\theta})$ the differential entropy rate of the output is given by~\cite[pp.~663]{Papoulis_Probability}
\begin{equation}
 \derate{Y} = \derate{X} + \frac{1}{2\pi}\int_{-\pi}^\pi \ln |G({\e{\jmath\theta}})| d \theta. \label{eq:dratediff}
\end{equation}

It can be shown (see, e.g.,~\cite{Yu_KLDDynSys}) that the integral above evaluates to $\ln |b_0|+\sum_{i:|z_i|>1}\ln|z_i|$, where $z_i$ are the zeros of the transfer function $G(z)$. For causal minimum-phase systems ($|z_i|<1\ \forall i$) with $b_0=1$, the differential entropy rates for the input and output process are equal. This result was recently verified by~\cite{Dumitrescu_EntropyInvariance}, which analyzed the invariance of entropy rates for all-pole filters. Scaling the transfer function of such a filter such that $b_0\neq 1$ leads to $\derate{X} \neq \derate{Y}$, despite the fact that by scaling no information is lost. Conversely, it is easily possible that $\derate{X}=\derate{Y}$ for systems which destroy information. Therefore, we believe that differential entropies and differential entropy rates are not adequate measures for information loss. Future investigations will show if alternative descriptions for continuous-valued processes will yield more appropriate characterizations.

\section{Other Examples}
\label{sec:examples}
While the case of linear filters is a particularly interesting one, the restriction to countable input and output alphabets suggests further examples illustrating the application of our theoretical results.

\subsection{Example 1: Finite-Precision Linear Filters}
The first example considers an extension to the subclass of discrete-time linear filters discussed in Section~\ref{sec:linfilt}. In many practical applications in digital signal processing linear filters are implemented with finite-precision number representations only.
We thus assume that both input process and filter coefficients take values from a finite set. For example, $\dom{X}$ may be a finite subset of the rational numbers $\mathbb{Q}$, closed under modulo-addition. Multiplying two values from that set, e.g., by multiplying an input sample with a filter coefficient, typically yields a result not representable in $\dom{X}$. As a consequence, after every multiplication a quantizer is necessary, essentially truncating the additional bits resulting from multiplication. Let the quantizer be described by a function $Q{:}\; \mathbb{R}\to\dom{X}$ with $\quant{a+X_n}=\quant{a}\oplus X_n$ if $X_n\in\dom{X}$, where $\oplus$ denotes modulo-addition  (e.g.,~\cite[pp.~373]{Oppenheim_DiscreteSigProc}). With this~\eqref{eq:linfilt} changes to
\begin{equation}
 Y_n=\bigoplus_{k=0}^N \quant{b_kX_{n-k}} \oplus \bigoplus_{l=1}^M \quant{a_lY_{n-l}} \label{eq:nlinfilt}
\end{equation}
or
\begin{equation}
 Y_n=\quant{\bigoplus_{k=0}^N b_kX_{n-k} \oplus \bigoplus_{l=1}^M a_lY_{n-l}} \label{eq:nnlinfilt}
\end{equation}
depending whether quantization is performed after multiplication or after accumulation (in the latter case, the intermediate results are represented in a larger set $\dom{X}'$). Note that due to modulo-addition the result $Y_n$ remains in $\dom{X}$.

We will now focus on filters with $b_0=1$. For filters with infinite precision this can be done without loss of generality by considering a constant gain factor $b_0$ and by normalizing all $b_k$ coefficients. However, this gain normalization poses a restriction in the finite-precision case since $b_k/b_0$ is not necessarily an element of $\dom{X}$. With $b_0=1$~\eqref{eq:nlinfilt} and~\eqref{eq:nnlinfilt} change to
\begin{equation}
 Y_n= X_n \oplus \left(\bigoplus_{k=1}^N \quant{b_kX_{n-k}} \oplus \bigoplus_{l=1}^M \quant{a_lY_{n-l}}\right) \label{eq:nlinfilt2}
\end{equation}
and
\begin{equation}
 Y_n=X_n\oplus\quant{\bigoplus_{k=1}^N b_kX_{n-k} - \bigoplus_{l=1}^M a_lY_{n-l}} \label{eq:nnlinfilt2}
\end{equation}
by the property of the quantizer. From this it can be seen that either implementation is partially invertible (the terms in parentheses in~\eqref{eq:nlinfilt2} and~\eqref{eq:nnlinfilt2} are both in $\dom{X}$, and modulo-addition has an inverse element). Consequently, even filters with nonlinear elements can be shown to preserve information under certain circumstances despite the fact that the quantizer function is non-injective. 

\subsection{Example 2: Multiplying Consecutive Inputs}\label{ssec:multipliers}
Another nonlinear system satisfying Definition~\ref{def:partInv} is given by the following input-output relationship:
\begin{equation}
 Y_n=X_nX_{n-1}
\end{equation}
The partial inverse in this case would be $X_n=\frac{Y_n}{X_{n-1}}$ if $X_{n-1}\neq0$, while for $X_{n-1}=0$ no such inverse exists. Therefore, this example represents a class of systems whose partial invertibility depends on the alphabet $\dom{X}$ of the stochastic process. If the process $\Xvec$ is such that $\dom{X}$ does not contain the element $0$, the partial inverse exists and we obtain for $X_n$, $n>1$:
\begin{equation}
 X_n=\begin{cases}
      X_1 \prod_{k=1}^{\frac{n-1}{2}} \frac{Y_{2k+1}}{Y_{2k}},& \text{for odd $n$}\\
\frac{Y_n}{X_1} \prod_{k=1}^{\frac{n}{2}-1} \frac{Y_{2k}}{Y_{2k+1}},& \text{for even $n$}
     \end{cases}\label{eq:reconstructMult}
\end{equation}
Indeed, since all $X_n$, $n>1$, can be computed from $X_1$ and $Y_1^n$, we obtain $\ent{X_1^n|Y_1^n}= \ent{X_1|Y_1^n}$ which is in perfect accordance with Theorem~\ref{thm:InfoLossPI}. Reconstruction of $\Xvec$ is thus possible up to an unknown $X_1$. Note, however, that this unknown sample influences the whole reconstructed sequence as shown in~\eqref{eq:reconstructMult}. Thus, even though the information loss rate vanishes, perfect reconstruction of any subsequence of $\Xvec$ is impossible by observing the output process $\Yvec$ only.

\subsection{Example 3: Hammerstein Systems}\label{ssec:Hammerstein}
\begin{figure}[t!]
 \centering
\begin{pspicture}[showgrid=false](1,1.5)(7,2.5)
	\psset{style=Arrow, style=RoundCorners}
 	\pssignal(1,2){z}{$\Xvec$}
	\psblock(3,2){ap1}{$g(\cdot)$}
	\psblock(5,2){ap2}{\parbox{\psunit}{\centering Linear Filter}}
	\pssignal(7,2){y}{$\Yvec$}
  \nclist[style=Arrow]{ncline}[naput]{z,ap1,ap2 $\mathbf{V}$,y}
\end{pspicture}
 \caption{Discrete-Time Hammerstein System. For $g(\cdot)=(\cdot)^2$ and if the linear filter is a moving-average filter, this corresponds to a discretized model of the energy detector.}
 \label{fig:ed}
\end{figure}
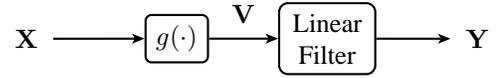

A final example considers a simple special case of a nonlinear dynamical system, namely, a cascade of a static nonlinearity and a linear filter~\cite{Shmaliy_CTSystems}. Such a cascade, usually referred to as Hammerstein system, is depicted in Fig.~\ref{fig:ed}. A practical example of such a Hammerstein system is the energy detector, a popular low-complexity receiver architecture in wireless communications. In the discrete-time case the input-output relationship is given by
\begin{equation}
 Y_n = \sum_{k=0}^N b_kg(X_{n-k}) + \sum_{l=1}^M a_lY_{n-l}.
\end{equation}
As it is easily seen, this system is partially invertible if and only if the function $g$ has an inverse. If $g$ is not invertible, we obtain in the light of Theorem~\ref{thm:UBInfoLossRate}:
\begin{equation}
 Y_n =f_{\Theta_n}(X_n)= b_0g(X_n) + C_{\Theta_n}
\end{equation}
where $C_{\Theta_n}$ is a constant depending on the random parameter $\Theta_n$. With this and $f^{-1}_\theta[f_\theta(x)]=g^{-1}[g(x)]$ for all $x\in\dom{X}$, $\theta\in\dom{T}$ we obtain an upper bound on the information loss rate:
\begin{IEEEeqnarray}{RCL}
 \entrate{X|Y}\leq\max_{x\in\dom{X}} \log|g^{-1}[g(x)]|
\end{IEEEeqnarray}

Interestingly, the structure of this system allows a simplified analysis: Since the information loss rate of a cascade of systems is equal to the sum of individual information loss rates (cf.~Theorem~\ref{thm:cascade}) we can analyze both constituent systems separately. The linear filter was already shown to preserve full information, so any information loss will be caused by the static nonlinearity, i.e., $\entrate{X|Y}=\entrate{X|V}$. This is in accordance with the observation that the Hammerstein system is partially invertible if the static nonlinearity is invertible.

For static nonlinearities the analytic treatment of information loss is simple compared to dynamical systems. In particular, for an independent, identically distributed (iid) input process $\Xvec$ the information loss rate can be shown to be equal to the zeroth-order conditional entropy, $\ent{X|V}$, while for a general stationary process this quantity acts as an upper bound~\cite{Watanabe_InfoLoss}. The upper bound from Theorem~\ref{thm:UBInfoLossRate} turns out to be even more general, since it also provides an upper bound on $\ent{X|V}$ in the case of an iid input process (cf.~Theorem~4 in~\cite{Geiger_ISIT2011arXiv}). An in-depth analysis of the interplay between these bounds is the object of future work.

\section{Conclusion}
In this work we have presented general results on the information loss of dynamical systems for stationary stochastic input and output processes on countable alphabets. Furthermore, we have extended the proof of the data processing inequality stating that the entropy rate at the output of the system cannot be larger than the entropy rate at the input and have derived an upper bound on the information loss rate. The additivity of information loss rates for cascaded systems could be shown, too.

We have further identified a family of systems for which this upper bound is zero, i.e., for which the information loss rate vanishes. Not only linear filters belong to that family, but also their nonlinear counterparts common in finite-precision signal processing.

Future research will extend these results to the case of continuous-valued stochastic processes and the application to common nonlinear systems, e.g., Volterra models.

\section*{Acknowledgments}
The authors gratefully acknowledge discussions with Sebastian Tschiatschek concerning mathematical notation, and his comments improving the quality of this manuscript.

\bibliographystyle{IEEEtran}
\bibliography{IEEEabrv,InformationProcessing,%
ProbabilityPapers,%
textbooks,%
myOwn,%
InformationWaves,%
ITBasics,%
HMMRate,%
ITAlgos}

\end{document}

%% file: abbrevations_processing.tex
\newcommand{\x}[1]{x[#1]}
\newcommand{\y}[1]{y[#1]}

\newcommand{\pdfy}{f_Y(y)}

\newcommand{\ent}[1]{H(#1)}
\newcommand{\diffent}[1]{h(#1)}
\newcommand{\derate}[1]{\bar{h}\left(\mathbf{#1}\right)}
\newcommand{\mutinf}[1]{I(#1)}
\newcommand{\ginf}[1]{I_G(#1)}
\newcommand{\kld}[2]{D(#1||#2)}
\newcommand{\binent}[1]{H_2(#1)}
\newcommand{\binentneg}[1]{H_2^{-1}\left(#1\right)}
\newcommand{\entrate}[1]{\bar{H}(\mathbf{#1})}
\newcommand{\mutrate}[1]{\mutinf{\mathbf{#1}}}
\newcommand{\redrate}[1]{\bar{R}(\mathbf{#1})}
\newcommand{\pinrate}[1]{\vec{I}(\mathbf{#1})}
\newcommand{\lossrate}[1]{L(\mathbf{#1})}

\newcommand{\dom}[1]{\mathcal{#1}}
\newcommand{\indset}[1]{\mathbb{I}\left({#1}\right)}

\newcommand{\unif}[2]{\mathcal{U}\left(#1,#2\right)}
\newcommand{\chis}[1]{\chi^2\left(#1\right)}
\newcommand{\chir}[1]{\chi\left(#1\right)}
\newcommand{\normdist}[2]{\mathcal{N}\left(#1,#2\right)}
\newcommand{\Prob}[1]{\mathrm{Pr}(#1)}
\newcommand{\Mar}[1]{\mathrm{Mar}(#1)}
\newcommand{\Qfunc}[1]{Q\left(#1\right)}

\newcommand{\expec}[1]{\mathrm{E}\left\{#1\right\}}
\newcommand{\expecwrt}[2]{\mathrm{E}_{#1}\left\{#2\right\}}
\newcommand{\var}[1]{\mathrm{Var}\left\{#1\right\}}
\renewcommand{\det}[1]{\mathrm{det}\left\{#1\right\}}
\newcommand{\cov}[1]{\mathrm{Cov}\left\{#1\right\}}
\newcommand{\sgn}[1]{\mathrm{sgn}\left(#1\right)}
\newcommand{\sinc}[1]{\mathrm{sinc}\left(#1\right)}
\newcommand{\e}[1]{\mathrm{e}^{#1}}
\newcommand{\multint}{\iint{\cdots}\int}
\newcommand{\modd}[3]{((#1))_{#2}^{#3}}
\newcommand{\quant}[1]{Q\left(#1\right)}

\newcommand{\hvec}{\mathbf{h}}
\newcommand{\avec}{\mathbf{a}}
\newcommand{\fvec}{\mathbf{f}}
\newcommand{\vvec}{\mathbf{v}}
\newcommand{\xvec}{\mathbf{x}}
\newcommand{\Xvec}{\mathbf{X}}
\newcommand{\Xhvec}{\hat{\mathbf{X}}}
\newcommand{\xhvec}{\hat{\mathbf{x}}}
\newcommand{\xtvec}{\tilde{\mathbf{x}}}
\newcommand{\Yvec}{\mathbf{Y}}
\newcommand{\yvec}{\mathbf{y}}
\newcommand{\Zvec}{\mathbf{Z}}
\newcommand{\wvec}{\mathbf{w}}
\newcommand{\Wvec}{\mathbf{W}}
\newcommand{\Hmat}{\mathbf{H}}
\newcommand{\Amat}{\mathbf{A}}
\newcommand{\Fmat}{\mathbf{F}}

\newcommand{\zerovec}{\mathbf{0}}
\newcommand{\eye}{\mathbf{I}}
\newcommand{\evec}{\mathbf{i}}

\newcommand{\zeroone}{\left[\begin{array}{c}\zerovec^T\\ \eye\end{array} \right]}
\newcommand{\zerooneT}{\left[\begin{array}{cc}\zerovec & \eye\end{array} \right]}
\newcommand{\zerooneM}{\left[\begin{array}{cc}\zerovec &\zerovec^T\\\zerovec& \eye\end{array} \right]}

\newcommand{\Cxx}{\mathbf{C}_{XX}}
\newcommand{\Cxh}{\mathbf{C}_{\hat{X}\hat{X}}}
\newcommand{\rxx}{\mathbf{r}_{XX}}
\newcommand{\Cxy}{\mathbf{C}_{XY}}
\newcommand{\Cyy}{\mathbf{C}_{YY}}
\newcommand{\Cnn}{\mathbf{C}_{NN}}
\newcommand{\Cyx}{\mathbf{C}_{YX}}
\newcommand{\Cygx}{\mathbf{C}_{Y|X}}

\newcommand{\NN}{{N{\times}N}}
\newcommand{\perr}{P_e}
\newcommand{\perh}{\hat{\perr}}
\newcommand{\pert}{\tilde{\perr}}

\newcommand{\vecind}[1]{#1_0^n}
\newcommand{\roots}[2]{{#1}_{#2}^{(i_{#2})}}
\newcommand{\rootx}[1]{x_{#1}^{(i_{#1})}}

\newcommand{\markkern}[1]{f_M(#1)}
\newcommand{\pole}{a_1}
\newcommand{\preim}[1]{g^{-1}[#1]}
\newcommand{\Xmax}{\bar{X}}
\newcommand{\Xmin}{\underbar{X}}
\newcommand{\xmax}{x_{\max}}
\newcommand{\xmin}{x_{\min}}
\newcommand{\limn}{\lim_{n\to\infty}}
\newcommand{\limX}{\lim_{\hat{\Xvec}\to\Xvec}}
\newcommand{\sumin}{\sum_{i=1}^n}
\newcommand{\finv}{f_\mathrm{inv}}

\newcommand{\modeq}[1]{g(#1)}

\newcommand{\delay}[2]{\psblock(#1){#2}{\footnotesize$z^{-1}$}}
\newcommand{\Quant}[2]{\psblock(#1){#2}{\footnotesize$\quant{\cdot}$}}
\newcommand{\moddev}[2]{\psblock(#1){#2}{\footnotesize$\modeq{\cdot}$}}